\documentclass{article}

\usepackage[final,nonatbib]{neurips_2021}

\usepackage{multirow}
\usepackage[utf8]{inputenc} 
\usepackage[T1]{fontenc}    
\usepackage[linkcolor=blue]{hyperref}
\usepackage{url}            
\usepackage{booktabs}       
\usepackage{amsfonts}       
\usepackage{nicefrac}       
\usepackage{microtype}      
\usepackage{xcolor}         

\usepackage{algorithm}

\usepackage[utf8]{inputenc} 
\usepackage[T1]{fontenc}
\usepackage{url}
\usepackage{ifthen}
\usepackage[cmex10]{amsmath} 

\usepackage{mathtools}

\usepackage{amsthm}
\newtheorem{proposition}{Proposition}
\newtheorem{definition}{Definition}

\usepackage{subfig}

\usepackage{float}
\usepackage{epsfig}
\usepackage{amsmath}
\usepackage{amsthm}
\usepackage{amssymb}

\usepackage{algpseudocode}
\usepackage{mathrsfs}
\usepackage{graphicx}
\usepackage[numbers,sort,compress]{natbib}
\usepackage{caption}
\usepackage{color}
\usepackage{chapterbib}
\usepackage{bbm}
\usepackage[noabbrev,capitalize,nameinlink]{cleveref}

\numberwithin{equation}{section}
\DeclareMathSymbol{:}{\mathord}{operators}{"3A}

\newcommand{\rt}{\mathbin{:}} 

\newcommand{\eq}{\begin{equation*}}
\newcommand{\en}{\end{equation*}}
\newcommand{\eqa}{\begin{eqnarray*}}
\newcommand{\ena}{\end{eqnarray*}}
\newcommand{\eqn}{\begin{equation}}
\newcommand{\enn}{\end{equation}}
\newcommand{\be}{\begin{equation}}
\newcommand{\ee}{\end{equation}}
\newcommand{\eqan}{\begin{eqnarray}}
\newcommand{\enan}{\end{eqnarray}}

\newcommand{\nn}{\nonumber}

\newcommand{\vv}{ {\bf v} }

\newcommand{\pmat}{\begin{pmatrix}}
\newcommand{\pman}{\end{pmatrix}}
\graphicspath{../figures/}
\graphicspath{{../figures/}}

\title{Estimating the Unique Information \\ of Continuous  Variables}

\author{%
Ari Pakman 
\\
Columbia University
\And 
Amin Nejatbakhsh
\\
Columbia University
\And 
Dar Gilboa 
\\
Harvard University
\And 
Abdullah Makkeh
\\
Georg August University
\And 
Luca Mazzucato
\\
University of Oregon
\And 
Michael Wibral
\\
Georg August University
\And 
Elad Schneidman 
\\
Weizmann Institute
}

\begin{document}

\maketitle
\begin{abstract}
The integration and transfer of information from multiple sources to multiple targets is a core  motive of neural systems. 
The emerging field of partial information decomposition~(PID) provides a novel information-theoretic lens into these mechanisms by identifying synergistic, redundant, and unique contributions to the mutual information between one and several variables. While many works have studied 
aspects of PID for Gaussian and discrete distributions, 
the case of general continuous distributions is still uncharted territory. In this work we present a method for estimating 
the unique information in continuous distributions, 
for the case of one versus two variables.
Our method solves the associated optimization problem over the space of distributions with fixed bivariate  marginals by combining copula decompositions and techniques developed to optimize variational autoencoders. We obtain excellent agreement with known 
analytic results for Gaussians, and  illustrate the power of our new approach in several brain-inspired neural models. Our method is capable of recovering the effective connectivity of a chaotic network of rate neurons, and uncovers a complex trade-off between redundancy, synergy and unique information in recurrent networks trained to solve a generalized XOR~task.
\end{abstract}

\section{Introduction and background}

In neural systems, often multiple neurons are driven by one external event or stimulus; conversely multiple neural inputs can converge onto a single neuron. 
A natural question in both cases is how multiple variables hold information about the singleton variable. In their seminal work~\cite{williams2010nonnegative}, Williams and Beer
proposed an axiomatic extension of classic information theory to decompose the mutual information between multiple  source variables and a single target variable in a meaningful way. 
For the case of two sources $X_1,X_2$, 
their partial information decomposition (PID)
amounts to expressing the mutual information 
of $X_1,X_2$ with a target $Y$ as 
a sum of four non-negative 
terms,
{\small 
\begin{eqnarray}
\label{eq:consistency1}
I(Y \rt (X_1,X_2)) =
U(Y \rt X_1 \backslash X_2) + U(Y \rt X_2 \backslash X_1)
+ R(Y \rt (X_1,X_2)) + S(Y \rt (X_1,X_2)) \,,
\end{eqnarray}
}
\!\! corresponding to unique ($U_1$, $U_2$), redundant~($R$) and synergistic~($S$)
contributions, respectively.  
These terms should also obey the consistency equations
\begin{eqnarray}
I(Y \rt X_1) = R(Y \rt (X_1,X_2))+ U(Y \rt X_1 \backslash X_2) \,,
\label{eq:consistency2}
\\
I(Y \rt X_2) = R(Y \rt (X_1,X_2))+ U(Y \rt X_2 \backslash X_1) \,.
\label{eq:consistency3}
\end{eqnarray}
The PID has proved useful
in understanding information processing by distributed systems in a diverse array of fields including machine learning~\cite{tax2017partial,wollstadt2021rigorous}, earth science~\cite{goodwell2020debates} and cellular automata~\cite{flecker2011partial},
and particularly in neuroscience~\cite{wibral2015bits,timme2016high,wibral2017partial,pica2017quantifying, kay2019bayesian}, where notions of synergy 
and redundancy, traditionally considered mutually exclusive and distinguished by the sign of
\eqan 
\nn
\Delta &=&  I(Y \rt (X_1,X_2)) -  I(Y \rt X_1) -  I(Y \rt X_2) \,,
\\
&=& S(Y \rt (X_1,X_2)) - R(Y \rt (X_1,X_2)) \,,
\label{eq:delta}
\enan 
have long played a central role 
in the quest to understand how neural circuits  
integrate information from multiple 
sources~\cite{gat1999synergy,brenner2000synergy,schneidman2003synergy,quiroga2009extracting}. The novelty of the PID framework here is  in separating the measures of synergy and redundancy in~(\ref{eq:delta}).

The above abstract formulation of PID  provides three equations for four unknowns, and only becomes operational once one of $U_1$, $U_2$, $R$, or $S$ is defined. This has been done in \cite{bertschinger2014quantifying} via a definition of the unique information:
\vskip .3cm 
\begin{definition}[BROJA~\cite{bertschinger2014quantifying}]
\label{def:broja}
Given three random variables $(Y,X_1,X_2)$
with joint probability density~$p(y,x_1, x_2)$,
the unique information $U_1$ of $X_1$ with respect to $Y$ is 
\eqan 
U(Y \rt X_1 \backslash X_2) &=& \min_{q \in Q} I_{q}(Y \rt X_1|X_2) \,,
\label{eq:unique}
\\
&=& \min_{q \in Q} 
\int dy dx_1 dx_2  \, q(y,x_1,x_2) 
\log\left( \frac{q(y,x_1|x_2)}{q(y|x_2) q(x_1|x_2)} \right) \,,
\label{eq:unique2}
\enan
where 
\eqan 
Q = \{ q(y,x_1, x_2) \, | \, q(y,x_i) = p(y,x_i), i=1,2 \} \,.
\label{eq:Q_family}
\enan 
\end{definition}

In words, we minimize the conditional mutual information $I(Y \rt X_1|X_2)$ 
over the space of density functions that preserve the marginal densities $p(y,x_1)$ and $p(y,x_2)$. The above definition implies, along with  (\ref{eq:consistency2})-(\ref{eq:consistency3}), 
that~the unique and redundant information only depend on the marginals $p(y,x_1), p(y,x_2)$, and that the  synergy can only be estimated from the full $p(y,x_1,x_2)$.

The original definition in \cite{bertschinger2014quantifying} was limited to discrete random variables. Here, we show that the extension to continuous variables is well-defined and can be practically estimated.

{\bf Motivation from decision theory}~\cite{bertschinger2014quantifying}.
Consider for simplicity discrete variables.
A  decision maker DM$_1$ 
 can choose an action $a$  from a finite  set ${\cal A}$,
and receives a reward~$u(a,y)$ based on the selected 
action and the state $y$, which occurs with probability $p(y)$.
Notably, DM$_1$ has no knowledge of~$y$, but observes instead a 
random signal $x_1$  sampled from $p(x_1|y)$. 
Choosing the action maximizing the expected reward for each $x_1$,
his maximal expected reward is 
\eqan 
R_1 = \sum_{x_1} p(x_1) \max_{a|x_1} \sum_{y}p(y|x_1) u(a,y) \,.
\enan 
DM$_1$  is said to have no unique information about $y$ w.r.t. 
another decision maker DM$_2$ that observes $x_2 \sim p(x_2|y)$  -- if $R_2 \geq R_1$ for any
set ${\cal A}$, any distribution $p(y)$, and any reward function $u(a,y)$. 
A celebrated theorem by Blackwell~\cite{blackwell1951comparison,leshno1992elementary} states that such a generic advantage 
by DM$_2$ occurs iff there exist a 
stochastic matrix $q(x_1| x_2)$ which  satisfies
\eqan 
p(x_1|y) = \sum_{x_2} p(x_2|y) q(x_1| x_2)\,.
\label{eq:blackwell}
\enan 
But this occurs precisely when the unique information (\ref{eq:unique})
vanishes, since then there exists a joint distribution $q(y,x_1,x_2)$ in $Q$ for which $y \perp x_1 | x_2$, which implies
$q(x_1| x_2,y)=q(x_1| x_2)$, and thus~(\ref{eq:blackwell}) holds.
Similar results exist for continuous variables~\cite{torgersen1991comparison,le1996comparison}.
Thus the unique information from~\cref{def:broja} quantifies a departure from Blackwell's relation~(\ref{eq:blackwell}).

In this work we present a definition and a
method to estimate the BROJA unique information
for generic continuous probability densities.
Our approach is based on the 
observation that the constraints~(\ref{eq:Q_family})
can be satisfied  
with an appropriate copula parametrization,
and makes use of techniques developed to optimize variational autoencoders.
We only consider one-dimensional $Y, X_1, X_2$ 
for simplicity, but the method 
can be naturally extended 
to  higher dimensional cases.
In \cref{sec:related_works} we review related works, in \cref{sec:estimation} we present our method and~\cref{sec:examples} contains several illustrative examples.

\section{Related works}
\label{sec:related_works}
Partial information decomposition offers a solution to a repeated question that was not addressed by `classical’ information theory regarding the relations between two sources and a target~\cite{williams2010nonnegative}. From a mathematical perspective a 'functional definition' has to be made, meaning that such a definition should align with our intuitive notions. Yet, as shown in \cite{bertschinger2012shared}, not all intuitively desirable properties of a PID can be realized simultaneously. Thus, different desirable properties are chosen for distinct application scenarios. Thus, various proposals for decomposition measures are not seen as conflicting but as having different operational interpretations.
For example, the BROJA approach used here builds on desiderata from decision theory, while other approaches appeal to game theory~\cite{ince2017measuring} or the framework of Kelly gambling~\cite{finn2018pointwise}. Yet other approaches use arguments from information geometry~\cite{harder2013bivariate}. 
 Other approaches assume agents receiving potentially conflicting or incomplete information about the source variables for the purpose of inference or decryption 
(see e.g.~\cite{rauh2017secret,makkeh2021differentiable}). In ~\cite{gutknecht2021bits} the authors separate the specific operational interpretations of PID measures from the general structure of information decomposition.

The actual computation of the BROJA unique information is non-trivial,
even for discrete variables. 
Optimization methods exist for the latter case~\cite{banerjee2018computing,makkeh2017bivariate,makkeh2018broja}, and analytic solutions are only known when all the variables are univariate binary~\cite{rauh2019properties}. 
For continuous  probability densities, 
an earlier definition aligned with the BROJA measure was made by Barret~\cite{barrett2015exploration}, but only applies to Gaussian variables. For Barret's measure, an analytic solution is known when $p(y,x_1,x_2)$ is a three-dimensional Gaussian density~\cite{barrett2015exploration}, but does not generalize to higher dimensional Gaussians~\cite{schamberg2021partial}. 




\section{Bounding and estimating the unique information}
\label{sec:estimation}
We proceed in two steps. We first 
introduce a parametrization of the optimization space $Q$ in (\ref{eq:Q_family})
and then introduce and optimize an upper bound on the unique information.

\subsection{Parametrizing the optimization space with copulas}
To characterize the optimization space $Q$ in 
(\ref{eq:unique})-(\ref{eq:Q_family}), it is convenient to 
recall that according to Sklar's theorem~\cite{sklar1959fonctions}, any $n$-variate probability density can be expressed as
\eqan 
p(x_1 \ldots x_n) = p(x_1) \ldots p(x_n) c(u_1 \ldots u_n) \,,
\label{eq:sklar}
\enan 
where $p(x_i)$ is the marginal and  $u_i= F(x_i)$ is the 
CDF of each variable. 
The dependency structure among the variables 
is encoded in the function 
$c: [0,1]^n \rightarrow [0,1]$. This  is a {\it copula} density, 
a probability density on the unit hypercube with uniform marginals~\cite{Joe1997-pk},
\eqan 
\underset{[0,1]^{n-1}}{\int}
  \prod_{j=1, j\neq i}^n du_j \, c(u_1 \ldots u_n)= 1 \quad  \forall i  \,.
\enan 
Note that under univariate 
reparametrizations $z_i'=g(z_i)$, the 
$u_i$'s and the copula $c$ remain invariant.
For an overview of copulas in 
machine learning, see~\cite{elidan2013copulas}.

\begin{proposition}
Under the BROJA~\cref{def:broja} of  unique information,
all the terms of the partial information decomposition in (\ref{eq:consistency1})-(\ref{eq:consistency3}) 
are independent of the univariate marginals $p(x_1),p(x_2),p(y)$,
and only depend on the copula $c(u_y,u_1,u_2)$.
\end{proposition}
\begin{proof}
Expressing $q(y,x_1,x_2), q(x_1,x_2), q(y,x_2)$ 
via copula decompositions~(\ref{eq:sklar}), 
and changing variables as $du_y = q(y)dy$, etc.,
the objective function in~(\ref{eq:unique2}) becomes
\eqan 
I_{q}(Y \rt X_1|X_2) &=& 
\underset{[0,1]^3}{\int}
du_y du_1 du_2 \, c(u_y,u_1,u_2) \log 
\left(
\frac{c(u_y,u_1,u_2)}{ c(u_y, u_2)  c(u_1,u_2)}
\right) \,.
\label{eq:I_q}
\enan 
Note that the copula of any marginal distribution is the marginal of the copula:
\eqan 
c(u_y, u_2) = \underset{[0,1]}{\int} du_1 \, c(u_y,u_1,u_2) \,,
\qquad \qquad 
c(u_1, u_2) = \underset{[0,1]}{\int} du_y \, c(u_y,u_1,u_2) \,.
\enan 
Thus the optimization objective and
the unique information are independent of the univariate marginals. 
A similar result holds for  the  mutual information terms in 
the l.h.s. of (\ref{eq:consistency1})-(\ref{eq:consistency3}).\footnote{The connection between mutual information and copulas 
was discussed in ~\cite{calsaverini2009information, ma2011mutual}.} 
It follows that none of the PID terms in (\ref{eq:consistency1})-(\ref{eq:consistency3}) depend on the 
univariate marginals, and therefore all the PID terms  
are invariant under univariate  reparametrizations of~$(y,x_1,x_2)$. 
\end{proof}

In order to parametrize the 
optimization space $Q$ in (\ref{eq:Q_family}) using copulas, consider the factorization
\eqan
p(y,x_1, x_2) = p(x_1) p(y|x_1) p(x_2|y,x_1) \,.
\label{eq:full_density}
\enan 
Using  the copula decomposition (\ref{eq:sklar}) for $n=2$, the last two factors in~(\ref{eq:full_density}) can be expressed as 
\eqan 
p(y|x_1) &=& \frac{p(y,x_1)}{p(x_1)} = 
 \frac{p(y) p(x_1) c(y,x_1)}{p(x_1)}
 =
c(u_y, u_1) p(y) \,,
\label{eq:exp1}
\enan 
and similarly
\eqan 
p(x_2|y,x_1) &=& \frac{p(x_1, x_2|y)}{p(x_1|y)} \,,
\\
&=& 
c_{1,2|y}(u_{1|y}, u_{2|y} ) p(x_2|y) \,,
\label{eq:exp2}
\\
 &=& c_{1,2|y}(u_{1|y}, u_{2|y} ) c(u_y, x_{2}) p(x_2) \,,
\label{eq:exp3}
\enan 
where we defined the conditional CDFs, 
\eqan 
u_{i|y} =F(u_i|u_y) = \frac{\partial C(u_y,u_i)}{\partial u_y }  \qquad i=1,2
\enan 
and $C(u_y,u_i)$ is the CDF of $c(u_y,u_i)$. 
Note that the function $c_{1,2|y}(u_{1|y}, u_{2|y} )$ in (\ref{eq:exp2}) is
not the conditional copula $c(u_{1}, u_{2}|u_y)$, but rather the copula of the conditional $p(x_1, x_2|y)$.
Using expressions (\ref{eq:exp1}) and (\ref{eq:exp3}),
the full density (\ref{eq:full_density})  becomes
\eqan
p(y,x_1, x_2) &=& 
p(y) p(x_1) p(x_2)  
c(u_y,u_1,u_2) \,, 
\label{eq:full_p}
\enan 
where 
\begin{eqnarray}
\label{eq:full_c}
\begin{aligned}
c(u_y,u_1,u_2) & = c(u_y, u_1)\, c(u_y, u_2)  c_{1,2|y}(u_{1|y}, u_{2|y} ) \,.
\end{aligned}
\end{eqnarray}
This is a simple case of the 
pair-copula construction of multivariate distributions~\cite{bedford2001probability,aas2009pair,czado2010pair}, which allows to expand any $n$-variate copula as a product
of (conditional) bivariate copulas. 


\vskip .3cm

\begin{proposition}
The copula of the conditional, $c_{1,2|y}(\cdot, \cdot)$,
parametrizes the space 
$Q$ in (\ref{eq:Q_family}). 
\end{proposition}
\begin{proof}
Since $q(y,x_i) = p(y,x_i)$ $(i=1,2)$, the copula factors in
\eqan 
p(y,x_i) = p(y) p(x_i) \, c(u_y, u_i) \,, \qquad i=1,2
\enan
are fixed in $Q$. 
Therefore, in the copula decomposition (\ref{eq:full_c}) for 
$q(y,x_1, x_2) \in Q$, only the last factor can vary in $Q$.
Let us denote by $\theta$ the parameters of a generic parametrization for the copula $c_{1,2|y}(u_{1|y}, u_{2|y} )$. 
Since the latter is conditioned
on $u_y$, the parameters can be 
taken as a function~$\theta(u_y)$. It follows that the copula of $q$ necessarily has the form 
\begin{eqnarray}
\label{eq:q_expansion}
\begin{aligned}
c_{\theta}(u_y,u_1,u_2) & =c(u_y, u_1) \, c(u_y, u_2) \,
  c_{1,2|\theta(u_y)}(u_{1|y}, u_{2|y}) \,,
\end{aligned}
\end{eqnarray}
and the parameters of the function 
$\theta(u_y)$ 
are the optimization variables.\footnote{We note that 
in multivariate pair-copula expansions 
it is common to assume constant conditioning parameters~$\theta$~\cite{nagler2016evading}, 
but we do not make such a simplifying  assumption. } 
\end{proof}

\subsection{Optimizing an upper bound} 
Inserting now the expression (\ref{eq:q_expansion}) into the objective function (\ref{eq:I_q})
we get
\eqan 
\label{eq:I_q_theta}
I[\theta] 
=
\mathbb{E}_{c_{\theta}(u_y,u_1,u_2)}
\log \left[ 
c(u_y, u_1) c_{1,2|\theta(u_y)}(u_{1|y}, u_{2|y})
\right] 
- \mathbb{E}_{c_{\theta}(u_1,u_2)}
\log  c_{\theta}(u_1,u_2) \,,
\enan 
which is our objective function and satisfies the marginal constraints (\ref{eq:Q_family}). 
Note that apart from the optimization parameters $\theta$,
it depends on the bivariate copulas 
$c(u_y, u_1)$ and $c(u_y, u_2)$ which should be estimated from the observed data. Given $D$ observations 
$\{ y^{(i)}, x_1^{(i)},x_2^{(i)})\}_{i=1}^D$,
we map each value to $[0,1]$
via the empirical CDFs of each coordinate $(y,x_1,x_2)$.
Computing the latter has a $O(D \log D)$ cost from sorting each coordinate 
and yields a data set $\{ u_y^{(i)}, u_1^{(i)},u_2^{(i)})\}_{i=1}^D$. 
The latter set is used 
to estimate copula densities $c(u_y, u_1)$ and $c(u_y, u_2)$ by fitting several  parametric and non-parametric copula 
models~\cite{Nelsen2007-ch}, 
and choosing the best pair of models using the AIC criterion.\footnote{For this fitting/model selection step, we used the 
\texttt{pyvinecopulib} python package~\cite{nagler_thomas_2020_4288293}.  }
From the learned  copulas 
we also get the conditional CDF functions  $u_{i|y}= F(u_i|u_y)$ 
that appear in the arguments of the first term in (\ref{eq:I_q_theta}).

{\bf A variational upper bound. }
Minimizing (\ref{eq:I_q_theta}) directly 
w.r.t. $\theta$ is 
challenging because the second term 
depends on the copula marginal 
$c_{\theta}(u_1,u_2)$ which 
has no closed form, as it requires 
integrating (\ref{eq:q_expansion}) w.r.t. $u_y$. 
We  introduce instead an inference distribution $r_{\phi}(u_y|u_1,u_2)$, with parameters $\phi$, 
that approximates the conditional copula 
$c_{\theta}(u_y|u_1,u_2)$, and consider the bound 
\eqan 
\label{eq:lower_bound}
\log c_{\theta}(u_1,u_2) = \log \int du_y' \, c_{\theta}(u_y',u_1,u_2) 
\geq
\int du_y' \, r_{\phi}(u_y'|u_1,u_2)  \, \log \frac{c_{\theta}(u_y',u_1,u_2)}{r_{\phi}(u_y'|u_1,u_2)} \,,
\enan 
which follows from Jensen's inequality and is tight when 
$r_{\phi}(u_y'|u_1,u_2)= c_{\theta}(u_y'|u_1,u_2)$. 
This expression gives an upper bound on $I_q[\theta]$,
which can be minimized jointly w.r.t. $(\theta, \phi)$. 

A disadvantage  of the bound (\ref{eq:lower_bound})
is that its tightness depends strongly on 
the expressiveness of the inference distribution 
$r_{\phi}(u_y' |u_1,u_2)$. 
This situation can be improved by considering a 
multiple-sample generalization proposed by~\cite{burda2015importance},
\eqan 
\log c_{\theta}(u_1,u_2) \geq D_{A,\theta, \phi}(u_1,u_2) 
\equiv \mathbb{E}_{p(u_y^{(1)} \ldots u_y^{(A)}) } 
\log \left[
\frac{1}{A} \sum_{a=1}^A 
\frac{c_{\theta}(u_y^{(a)},u_1,u_2)}{r_{\phi}(u_y^{(a)}|u_1,u_2)}
\right] \,, 
\label{eq:Dk_leq}
\enan 
where the expectation is w.r.t. $A$ independent samples of $r_{\phi}(u_y' |u_1,u_2)$. 
$D_{A,\theta, \phi}(u_1,u_2)$ coincides with the lower bound in~(\ref{eq:lower_bound}) for $A=1$ and satisfies~\cite{burda2015importance}
\eqan 
D_{A+1,\theta, \phi}(u_1,u_2) &\geq & D_{A,\theta, \phi}(u_1,u_2),
\\
\lim_{A\rightarrow \infty} D_{A,\theta, \phi}(u_1,u_2) &=& \log c_{\theta}(u_1,u_2) \,.
\enan 
Thus, even when $r_{\phi}(u_y'|u_1,u_2)\neq c_{\theta}(u_y'|u_1,u_2)$,
the bound  can be made arbitrarily tight
for large enough~$A$. 
Inserting (\ref{eq:Dk_leq}) in (\ref{eq:I_q_theta}), we get finally
\eqan 
I_q[\theta] \leq 
B_1[\theta]  + B_2[\theta,\phi] \,, 
\label{eq:I_q_bound}
\enan 
where 
\eqan 
B_1[\theta] &=&
\mathbb{E}_{c_{\theta}(u_y,u_1,u_2)}
\log \left[ 
c(u_y, u_1) 
c_{1,2|\theta(u_y)}(u_{1|y}, u_{2|y})
\right] \,, 
\\
B_2[\theta,\phi]& = &
- \mathbb{E}_{c_{\theta}(u_1,u_2)}
 D_{A,\theta, \phi}(u_1,u_2) \,,
\enan 
and we minimize the r.h.s. of (\ref{eq:I_q_bound}) w.r.t. $(\theta,\phi)$.
Low-variance estimates of the gradients to perform the minimization can be obtained 
with the reparametrization trick~\cite{kingma2013auto,tucker2018doubly}, 
as discussed in detail in Appendix~\ref{app:gradients}. 
In our examples below  we use for $c_{1,2|\theta(u_y)}$
a bivariate Gaussian copula
 (reviewed in Appendix~\ref{app:gaussian_copula}).
Such a copula has just one parameter $\theta \in [-1,+1]$, and thus the optimization is  done 
over the space of functions 
$\theta(u_y):[0,1] \rightarrow [-1,+1]$,
which we parametrize with a two-layer neural network. 
Similarly,  we parametrize
$r_{\phi}(u_y|u_1,u_2)$ with a two-layer neural 
network. Details of these networks are in  Appendix~\ref{app:inference}.

While the term $B_2$ in our bound is similar to the negative of the ELBO bound 
in importance weighted autoencoders (IWAEs)~\cite{burda2015importance}, there are some differences
between the two settings, the most important being that we are 
interested in the precise value 
of the bound at the minimum, rather than the learned functions $c_{\theta}, r_{\phi}$. Note also that our latent variables $u_y^{(k)}$ are one-dimensional,
as opposed to the usual  higher dimensional latent distributions 
of variational autoencoders, and that 
the empirical expectation over data observations in IWAEs
is replaced in $B_2$  by the expectation over~$c_{\theta}(u_1,u_2)$, 
whose parameters are also optimized.

\begin{figure*}[t!]
	\begin{center}
		\fbox{		
		\includegraphics[width=.95\textwidth]{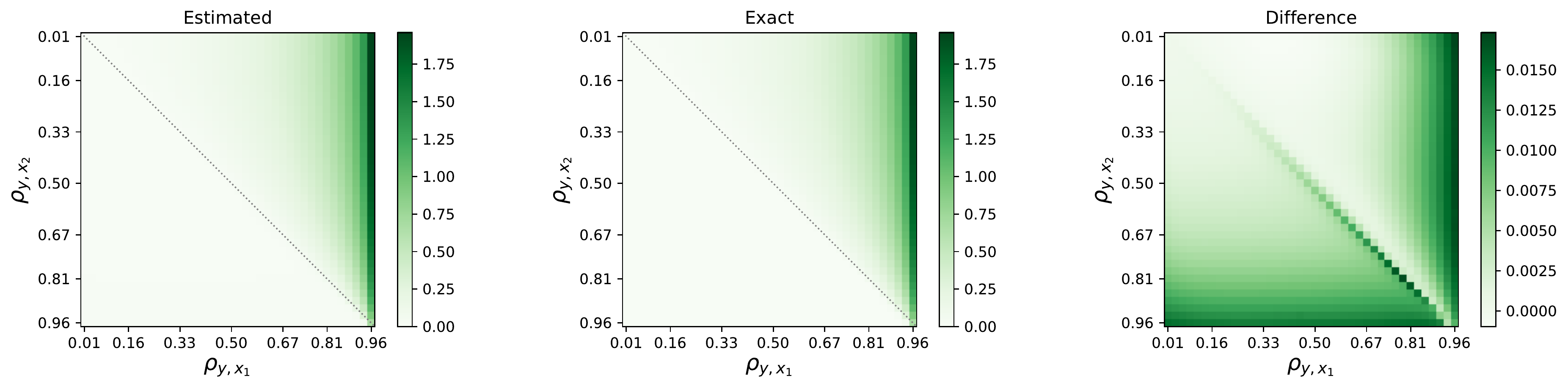}}
	\end{center}
	\caption{
	 {\bf Estimated vs. exact  values of unique information  for Gaussians.}
	 For a three-dimensional Gaussian, we show estimates of  $U(Y \rt X_1
	 \backslash X_2)$ 
	 as a function of the  correlations $\rho_{y,x_i} (i=1,2)$, compared with the exact results from~\cite{barrett2015exploration}. 
 Only for Gaussian distributions are exact results known for continuous variables.	 }  		\label{fig:comparison_Gaussian}	
\end{figure*}

{\bf Estimating the other  PID terms}
\label{sec:estimating_all}
In the following we adopt the minimal value taken by  the upper bound (\ref{eq:I_q_bound}) as our estimate of~$U_1$. The other terms in the partial information decomposition
are obtained 
from the consistency relations (\ref{eq:consistency1})-(\ref{eq:consistency3}), 
after estimating the mutual informations 
$I(Y \rt (X_1,X_2)), I(Y \rt X_1), I(Y \rt X_2)$.  
There are several methods for the latter. In our examples, we use the observed data to fit additional copulas $c(u_1,u_2)$ and $c_{12|\theta(u_y)}$ and 
estimate $I(Y \rt X_1) \simeq \frac{1}{D} \sum_{i=1}^D \log c(u_y^{(i)},u_1^{(i)})$
and similarly for the other terms. Note that all 
our estimates  have  sources of potential bias.
Firstly, the estimation of the parametric copulas 
is subject to model or parameter misspecification,
which can be ameliorated by more refined model selection strategies.
Secondly, the optimized bound might not saturate, biasing the estimate upwards. This can be improved using higher $A$ values 
and improving the gradient-based optimizer used.




\section{Examples}

\label{sec:examples}
{\bf Comparison with exact results for Gaussians.}
Consider a three-dimensional Gaussian  
with correlations $\rho_{y,x_i}$ between $y,x_i$ for $i=1,2$.
The exact solution to~(\ref{eq:unique}) in this case is~\cite{barrett2015exploration}
\eqan 
U(Y \rt X_1 \backslash X_2)  = \ensuremath{\frac{1}{2}\log\left(\frac{1-\rho_{y,x_{2}}^{2}}{1-\rho_{y,x_{1}}^{2}}\right)}\mathbbm{1}\left[\ensuremath{\rho_{y,x_{2}}<\rho_{y,x_{1}}}\right].
\enan 
Fig.~\ref{fig:comparison_Gaussian} compares the above expression
with estimates from our method. 
Here we know that~$c_{y,1}$ and $c_{y,2}$
are Gaussian copulas, 
with parameters $\rho_{y,x_1},\rho_{y,x_2}$, 
and we assumed a Gaussian copula for 
$c_{1,2|y,\theta}(u_{1|y}, u_{2|y})$ as well. 
For each  pair of values $\rho_{y,x_1},\rho_{y,x_2}$.
In this and the rest of the  experiments, 
we optimized the parameters $(\theta,\phi)$ using the  ADAM algorithm~\cite{kingma2014adam}
with a fixed learning rate $10^{-2}$ during 1200 iterations,
and using $A=50$. The results reported correspond to the mean of the bound 
in the last 100 iterations. 
The comparison in Fig.~\ref{fig:comparison_Gaussian}
shows excellent agreement. 

{\bf Model systems of three neurons.}
The nature of information processing of neural systems
is a prominent area of application of 
the PID framework,
since synergy has been proposed as natural measure of information modification~\cite{lizier2013towards,timme2016high}.
We consider two models:
\eqan
\begin{array}{c}
{\bf M1}
\\
(X_{1},X_{2})\sim\mathcal{N}(0,\rho_{12}^{2}),
\\
Y=\tanh(w_{1}X_{1}+w_{2}X_{2}).
\end{array}
\,  
\begin{array}{c}
{\bf M2}
\\
(X_{1},X_{2})\sim\mathcal{N}(0,\rho_{12}^{2}),\\
 Y=X_{1}^2/\left(0.1+w_{1}X_{1}^2+w_{2}X_{2}^2\right).
 \end{array}
\label{eq:models12}
\enan 
Both models are parameterized by the correlation $\rho_{12}$ and weights $w_1,w_2$. 
Model 1 is a particularly simple neural network. The $\tanh$ activation does not affect its copula, and  even for a linear activation function the variables are not jointly Gaussian since $Y$ is  deterministic on $(X_1,X_2)$. Model 2 is inspired by a normalization operation widely believed to be canonical in neural systems~\cite{Carandini2011-ju} 
and plays a role 
in common learned image compression methods~\cite{balle2016end}.
The results, presented in~\cref{fig:PID}. are obtained from~$3000$ samples from each model

\begin{figure*}[t!]
\begin{center}
\fbox{
\includegraphics[width=.48\textwidth, height=1.5in]{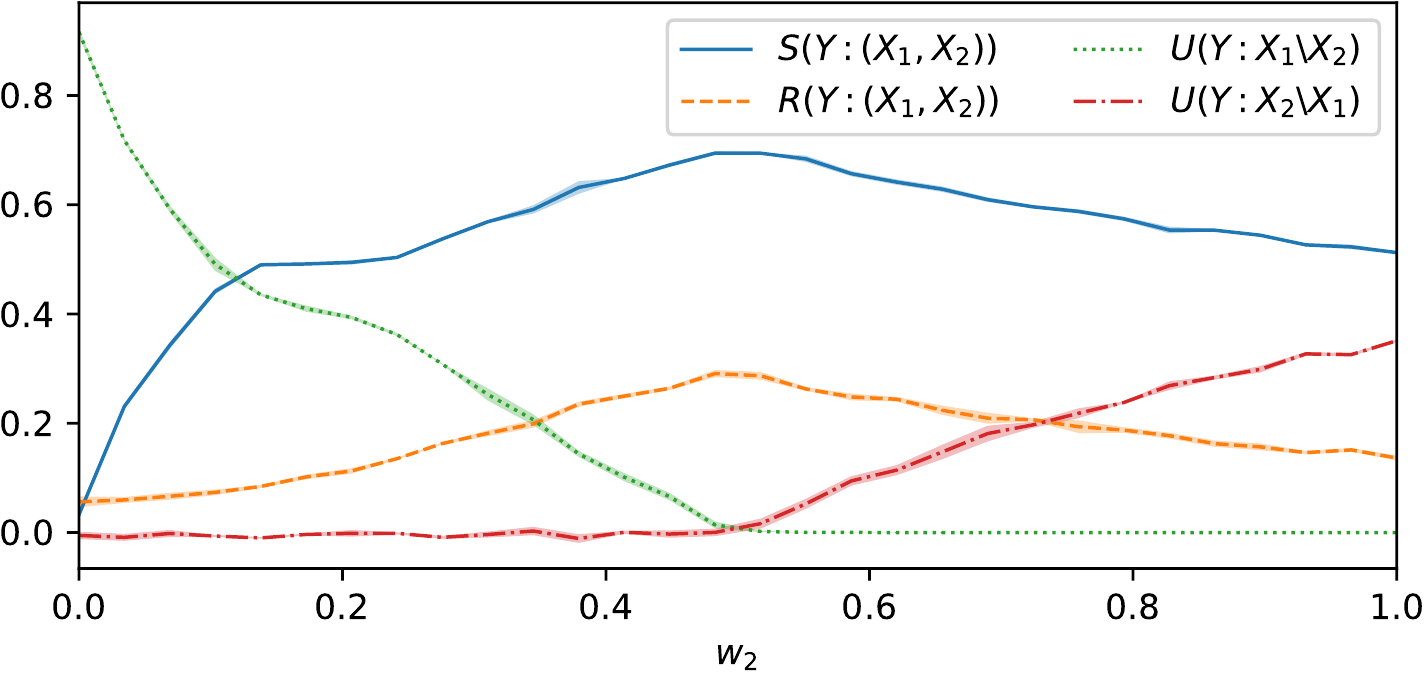}
\includegraphics[width=.48\textwidth, height=1.5in]{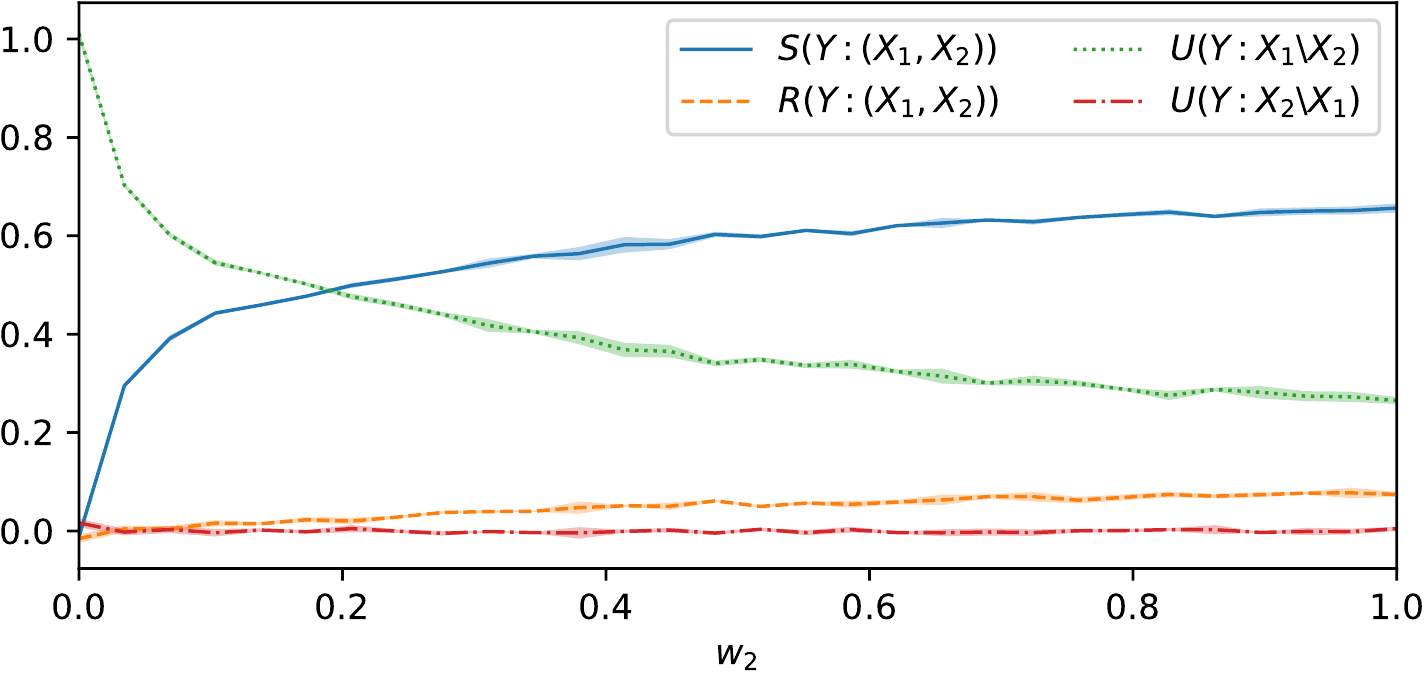}}
\end{center}
\caption{\textbf{Partial information decomposition for two neural network models.} 
In both models~\eqref{eq:models12} we fixed $w_1=0.5,\rho_{12}=0.3$,
and show the PID terms as a function of the synaptic strength~$w_2$, normalized by $I(Y \rt (X_1,X_2))$. We show mean (lines) and standard deviations (shaded area around each line) from 3 runs. 
\textit{Left:} Model 1: The input of greatest weight conveys 
all the unique information, and synergy and redundancy  both peak as $w_1=w_2$. 
\textit{Right:} Model 2: The second input $X_2$
has negligible  unique information contribution, but its
synaptic strength~$w_2$ modulates 
the synergistic term, associated to the
modification of information the neuron performs~\cite{lizier2013towards}.
}
\label{fig:PID}
\end{figure*}



{
{\bf Computational aspects of connectivity in recurrent neural circuits.}
{We apply our continuous variable PID to understand computational aspects of the information processing  between recurrently coupled neurons (Fig.~\ref{fig:conn}). A large amount of work has been devoted to applying information theoretic measures for quantifying directed pairwise information transfer between nodes in dynamic networks and neural circuits~\cite{reid2019advancing}. However, classical information theory only allows for the quantification of information transfer, whereas the framework of PID enables further decomposition of information processing into transfer, storage, and modification, providing further insights into the computation within a recurrent system~\cite{wibral2017quantifying}. Transfer entropy (TE) \cite{schreiber2000measuring} is a popular measure to estimate the directed transfer of information between pairs of neurons~\cite{vicente2011transfer,novelli2021inferring}, and is sometimes approximated by linear Granger causality.  Intuitively, TE between a process $X$ and a process $Y$ measures how much the past of $X$, $X^-$, can help to predict the future of $Y$, $Y^+$, accounting for its past $Y^-$. Although TE quantifies how much information is transferred between neurons, it does not shed light on the computation emerging from the interaction of $X^-$ and $Y^-$. Simply put, the information transferred from $X^-$ could enter $Y^+$, independently of the past state $Y^-$, or it could be fused in a non-trivial way with the information in the state in $Y^-$\cite{williams2011generalized,wibral2017quantifying}. 
PID decomposes the TE into \textbf{modified transfer} (quantified by $S(Y^+:X^-,Y^-)$) and \textbf{unique transfer} 
(quantified by $U(Y^+:X^- \setminus Y^-)$) terms (see the Appendix for a proof):
\begin{equation*}
TE(X\rightarrow Y)=I(Y^+: X^-|Y^-)=U(Y^+:X^-\setminus Y^-)+S(Y^+:X^-,Y^-)    \,.
\end{equation*}
Furthermore, the information kept by the system through time can be quantified by the \textbf{unique storage} (given by $U(Y^+:Y^- \setminus X^-)$) and \textbf{redundant storage} (given by $R(Y^+:X^-,Y^-)$) in PID~\cite{lizier2013towards}. This perspective is a new step towards understanding how the information is processed in recurrent systems beyond merely detecting the direction functional interactions estimated by traditional TE methods (see Appendix~\ref{apx:experiments}, for details).
To explore these ideas, we simulated chaotic networks of rate neurons with an a-priori causal structure consisting of two sub-networks $\mathbf{X}$ and $\mathbf{Y}$ (Fig.~\ref{fig:conn}a, see \cite{nejatbakhsh2020predicting} for more details on causal analyses of this network model). The sub-network $\mathbf{X}$ is a Rossler attractor of three neurons obeying the dynamical equations:
\begin{gather}
\begin{cases}
    \dot X_1 = -X_2-X_3\\
    \dot X_2 = X_1 + \alpha X_2\\
    \dot X_3 = \beta + X_3(X_1-\gamma)
\end{cases}
\end{gather}
where $\{\alpha,\beta,\gamma\}=\{0.2,0.2,5.7\}$. {There are 100 neurons in the sub-network $\mathbf Y$ from which we chose the first three, $Y_{1:3}$, to simulate the effect of unobserved nodes.} Neurons within the sub-network $Y$ obey the dynamical equations
\begin{gather}
    \dot Y= -\lambda Y + 10 \tanh (J_{YX}X+J_{YY}Y)
\end{gather}
where $J_{YX} \in \mathbb{R}^{100 \times 3}$ 
has all its entries equal to $0.1$, and $J_{YY}$ is the recurrent weight matrix of the $Y$ sub-network, sampled as zero-mean, independent Gaussian variables with standard deviation $g=4$. No projections exist 
from the downstream sub-network $\mathbf{Y}$ to the upstream sub-network $\mathbf{X}$. 
We simulated time series from this network (exhibiting chaotic dynamics, see Fig.~\ref{fig:conn}a) and estimated the PID as unique, redundant, and synergistic contribution of neuron $i$ and neuron $j$ at time~$t$ in shaping the future of neuron $j$ at time $t+1$. For each pair of neurons $Z_i,Z_j \in \{X_{1:3},Y_{1:3}\}$ we treated $(Z_i^t,Z_j^t,Z_j^{t+1})_{t=1}^T$ as iid samples\footnote{Note that the estimation of the PID from many samples of the triplets
$(Z_i^t,Z_j^t,Z_j^{t+1})$ is operationally the same whether such triplets  are iid or, as in our case, temporally correlated. This is similar to estimating expectations w.r.t. the equilibrium distribution  of a Markov chain  by using temporally correlated successive values of the chain.
In both cases, the temporal correlations do not introduce bias in the estimator but can increase the variance.} and ran PID on these triplets ($i,j$ represent rows and columns in Fig.~\ref{fig:conn}b-d). The PID uncovered the functional architecture of the network and further revealed non-trivial interactions between neurons belonging to the different sub-networks, encoded in four matrices: modified transfer~$S$, unique transfer~$U_1$, redundant storage~$R$, and unique storage~$U_2$~(details in Fig.~\ref{fig:conn}d). The sum of the modified and unique transfer terms was found to be consistent with the~TE (Fig.~\ref{fig:conn}c, TE equal to $S+U_1$, up to estimation bias). The TE itself captured the network effective connectivity, consistent with previous results \cite{novelli2021inferring,nejatbakhsh2020predicting}.
}

\begin{figure}[t!]
    \centering
    \includegraphics[width=.7\textwidth]{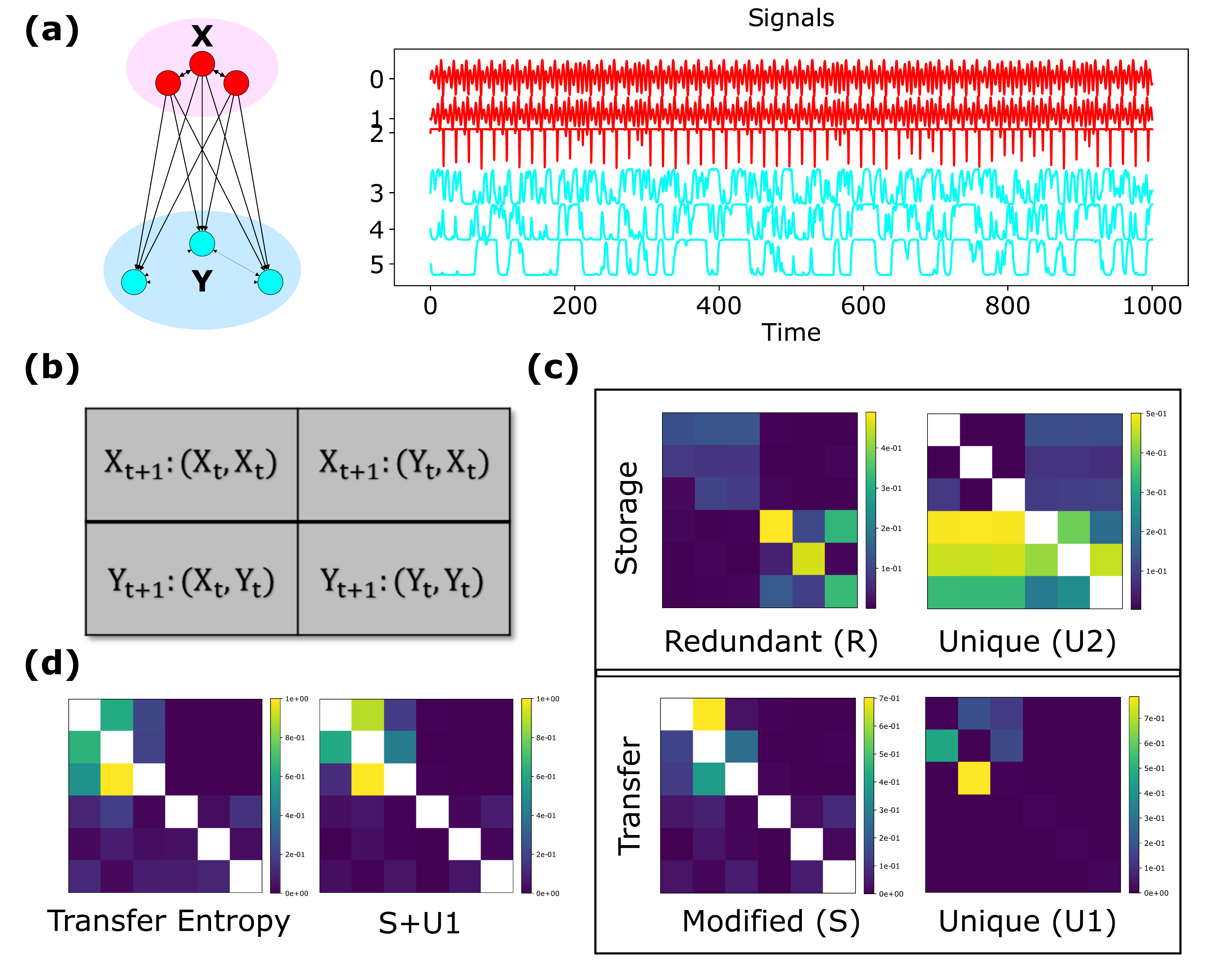}
    \caption{{\bf PID uncovers the effective connectivity {and allows for the quantification of storage, modification, and transfer of information in a chaotic network of rate neurons. }} {\bf a}: Schematics of recurrent network architecture (left) and representative activity (right). {\bf b}: Schematic of the PID triplets for each $3 \times 3$ block of the matrices in c, d. {\bf c:} PID decomposition into modified transfer $S$, unique transfer $U_1$, redundant storage $R$, and unique storage $U_2$ for the rate network. The future of $X$ neurons only depends on unique information in the past of $X$ neurons and their synergistic interactions. The interactions between the $X$ and $Y$ sub-networks only contain synergistic information regarding the future of $Y$ but no redundant information; the latter is only present in the interactions confined within each sub-network.
    {\bf d}: The transfer entropy (TE), estimated via IDTxl~\cite{wollstadt2019idtxl}, recovers the sum of  modified and unique transfer terms~$S+U_1$.}
    \label{fig:conn}
\end{figure}

\begin{figure}[t!]
    \centering
    \includegraphics[width=1\textwidth]{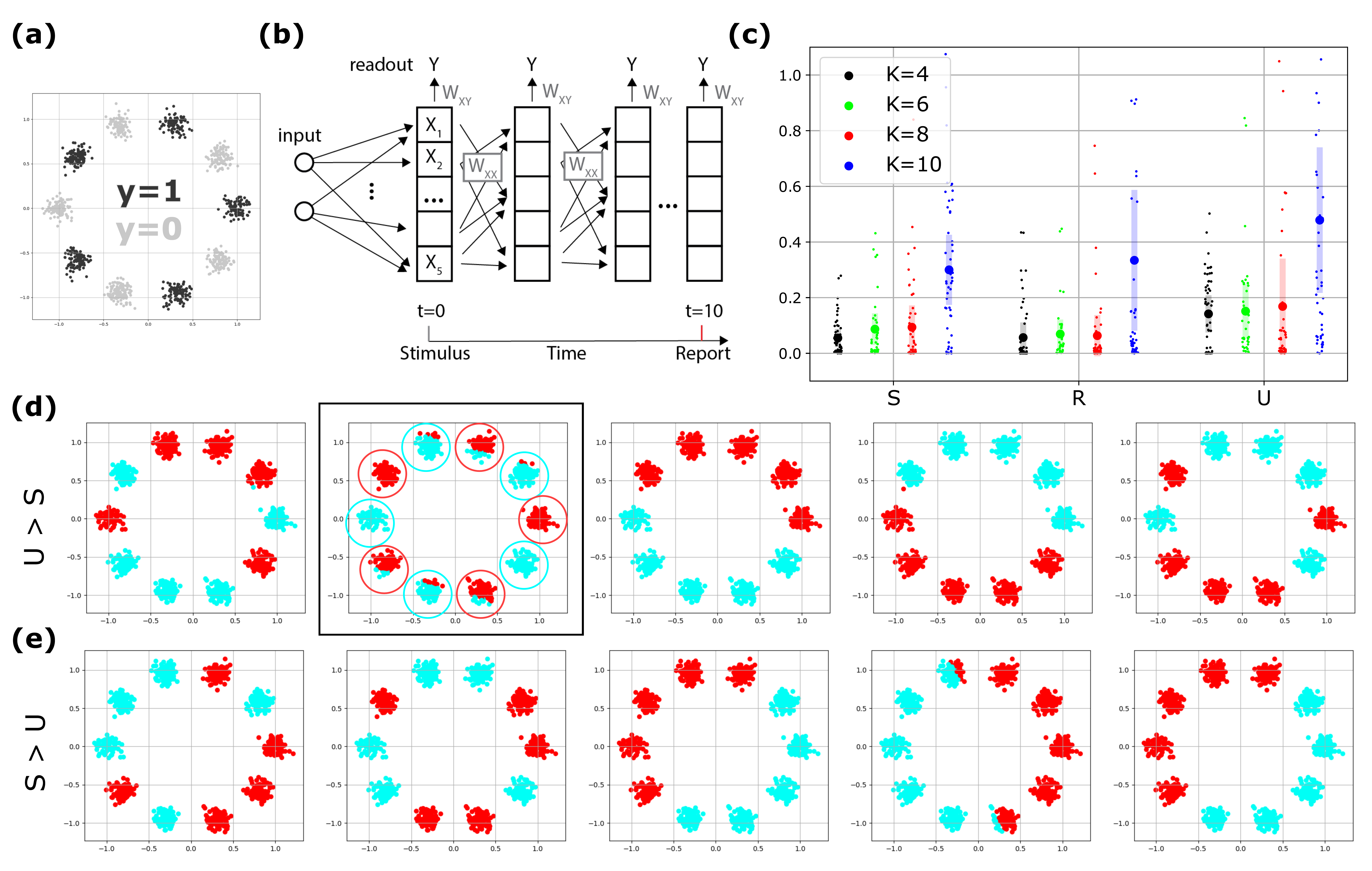}
    \caption{{\bf PID of RNNs trained to solve generalized XOR problem.} {\bf a}: Input data drawn from a 2D Gaussian Mixture Model with $K$ mixture components $X \sim \sum_{k=1}^K \frac{1}{K} \mathcal{N}(X|\mu_k,\sigma I)$ with means 
    lying on the unit circle (grey and black dots represent the two class labels). {\bf b}: Two layer network with 2D input layer, $5$ recurrently connected hidden neurons $X$ and one readout neuron $Y$; RNN activity unfolds in time (horizontal axis). The input is presented at time $t=0$, then withdrawn, and the RNN is trained with BPTT to report the decision at $t=10$. In this representation, layers correspond to time-steps and weights $W_{XX}$ are shared between layers.
    {\bf c}: PID between output $Y(t)$ and pairs of hidden neurons $X_i(t-1),X_j(t-1)$ for $t=10$ yielding $S,R,U_1,U_2$ (distribution over 1000 input samples for each task $K$; 20 networks per task). Harder tasks led to an increase in PID measures. {\bf d}: Example receptive fields for a network with $U>S$ shows emergence of grand-mother cells in the hidden layer (red and blue colors represent hidden neurons outputs; grandmother cell, second from left). {\bf e}: Example receptive fields for a network with $S>U$, relying on higher synergy between neurons to solve the task.}
    \label{fig:xor}
\end{figure}
 
{\bf Uncovering a plurality of computational strategies in RNNs trained to solve complex tasks.}
A fundamental goal in neuroscience is to understand the computational mechanisms emerging from the collective interactions of recurrent neural circuits leading to cognitive function and behavior. Here, we show that PID opens a new window for assessing how specific computations arise from recurrent neural interactions. Unlike MI or TE, the PID quantifies the alternative ways in which a neuron determines the information in its output from its inputs, and thus can be a sensitive marker of different computational strategies.
We here trained RNNs as models of cortical circuits \cite{mante2013context} and used the PID to elucidate how the computations emerging from recurrent neural interactions contribute to task performance. We trained RNNs to solve a generalized version of the classic XOR classification problem with target labels corresponding to odd vs. even mixture components (Fig.~\ref{fig:xor}a). Stimuli were presented for one time step ($t=0$) and the network was trained to report the decision at $t=10$. By tracking the temporal trajectories of the hidden layer activity we found that the network recurrent dynamics (represented as unfolded in time in Fig.~\ref{fig:xor}b) progressively pulls the two input classes in opposite directions along the output weights (see Appendix). We used PID to dissect how a plurality of different strategies emerge from recurrent neural interactions in RNNs trained for solving a classification task. The computation emerged from the recurrent interaction between hidden neurons at different time steps. Do all successfully trained networks have a similar profile in terms of the PID terms? If so, this hints at a single computational strategy across these networks. If not, it is safe to assume that task performance is reached via different mechanisms, despite identical network architecture and training algorithm.


We found that on average across multiple networks S, R, and U rose with task difficulty (Fig. \ref{fig:xor}c), yet at all difficulties, individual networks differed strongly with respect to the ratio $S/U$, i.e. there were networks with larger average synergy across neuron pairs compared to the average unique information, and vice versa. For simple networks like the ones used here, one can inspect receptive fields to understand the reason for this differential behaviour (Fig. \ref{fig:xor}d-e). Indeed, networks with high average unique information displayed 'grandmother-cell'-like neurons, that would alone classify a large parts of the sample space, while in networks with higher average synergy such cells were absent (Fig. \ref{fig:xor}d). The emergence of these 'grandmother-cell'-like receptive fields is due to the recurrent dynamics. While in a feedforward architecture ($W_{XX}=0$) hidden layer receptive fields are captured by hyperplanes in input space, in the RNN the receptive fields are time dependent, where later times are interpreted as deeper layers (Fig.~\ref{fig:xor}b) and thus can capture highly non-linear features in input space.
The advantage of PID versus a manual inspection of receptive fields is twofold: First, the PID framework abstracts and generalizes descriptions of receptive fields as being e.g.\,'grandmother-cell'-like; thus the concept of unique information stays relevant even in scenarios where the concept of a receptive field becomes meaningless, or inaccessible. Second, the quantitative outcomes of a PID rest only on information theory, not specfic assumptions about neural coding or computational strategies, and can be obtained for large numbers of neurons.


Comparison of our PID-based approach with the concept of neuronal selectivity used in neuroscience highlights interesting similarities and differences. Several kinds of selectivity (pure, mixed linear, and mixed non-linear) can be identified by performing regression analysis of neural responses vs. task variables \cite{rigotti2013importance}. In this framework, our grand-mother cells correspond to neurons with pure selectivity to the input class labels (a.k.a. "choice-selective" neurons). In the XOR task,  \cite{rigotti2013importance} showed that non-linear mixed selectivity of neurons to the class labels is beneficial when solving the XOR task, by leading to a high-dimensional representation of the task variables. While selectivity profiles are a property of single neuron responses to task variables, our PID measures are a property of the combined activity of triplets of neurons and thus reveal emerging functional interactions between units and their computational algorithms (see also \cite{timme2016high} and \cite{wibral2017quantifying}). This allowed us to characterize a functional property of neural systems less studied than task variable selectivity: the computations that require functional mixing of the information from multiple units (measured by the average synergistic information) vs. the computations that rely on the output of individual neurons (measured by the unique information and described as grandmother cells). Concretely, by comparing PID and receptive fields we found that that in networks with high unique information, neurons typically have receptive fields with pure selectivity (grandmother cells, with large unique information to the class labels). In networks with high synergy, neurons show complex mixed selectivity to class labels.

}

\section{Conclusions}
We presented a partial information decomposition measure for continuous variables with arbitrary probability densities, thereby extending the popular BROJA PID measure for discrete variables. Extending PID measures to continuous variables drastically broadens the possible applications of the PID framework. This is important as the latter provides key insights into the way a complex system represents and modifies information in a computation -- via asking which variables carry information about a target uniquely (such that it can only be obtained from that variable), redundantly, or only synergistically with other variables. Answering these questions is pivotal to understanding distributed computation in complex systems in general, and neural coding in particular.
We believe that the methods presented here will allow PIDs to be extended efficiently in neuroscience for multiple continuous sources with potentially complex dependency structures, as would be common in cellular imaging data or activation properties of brain modules or areas in functional imaging. More generally, the approach we presented here would be relevant for other application domains such as machine learning, biomedical science, finance, and the physical sciences.


\newpage 
\clearpage

\section*{Acknowledgments}
We thank Thibault Vatter and Praveen Venkatesh for conversations. 
 The work of AP is supported by the Simons Foundation, the DARPA NESD program, NSF NeuroNex Award DBI1707398 and The Gatsby Charitable Foundation. DG is supported by a Swartz Fellowship.  AM and MW are supported by Volkswagenstiftung under the program `Big Data in den Lebenswissenschaften’ and by the Ministry for Science and Education of Lower Saxony and the Volkswagen Foundation through the `Niedersächsisches Vorab'.
 LM is supported by NINDS Grant NS118461 (BRAIN Initiative). ES is supported by the Simons Collaboration on the Global Brain (542997) as well as research support from Martin Kushner Schnur and Mr. and Mrs. Lawrence Feis, and is the Joseph and Bessie Feinberg Professorial Chair. 

\bibliographystyle{unsrt}  
\bibliography{thebib}

\newpage 
\appendix 
\section*{Supplementary Material}

\section{Estimating the gradients}
\label{app:gradients}
As shown in \cref{sec:estimation}, the unique information 
$U(Y \rt X_1 \backslash X_2)$ is upper bounded as 
\eqan 
I_q[\theta] \leq 
B_1[\theta]  + B_2[\theta,\phi] \,, 
\label{app_eq:I_q_bound}
\enan 
where 
\eqan 
B_1[\theta] &=&
\mathbb{E}_{c_{\theta}(u_y,u_1,u_2)}
\log \left[ 
c(u_y, u_1) c_{1,2|\theta(u_y)}(u_{1|y}, u_{2|y})
\right] \,, 
\label{eq:B1}
\\
B_2[\theta,\phi]& = &
- \mathbb{E}_{c_{\theta}(u_1,u_2)}
 D_{A,\theta, \phi}(u_1,u_2) \,,
\label{eq:B2}
\enan 
and
\eqan 
D_{A,\theta, \phi}(u_1,u_2) = 
\mathbb{E}_{r_{\phi}(u_y^{(1)} \ldots u_y^{(A)}|u_1,u_2 ) } 
\log \left[
\frac{1}{A} \sum_{a=1}^A 
\frac{c_{\theta}(u_y^{(a)},u_1,u_2)}{r_{\phi}(u_y^{(a)}|u_1,u_2)}
\right] \,, 
\label{eq:Dk}
\enan 
and the above expectation is w.r.t.
\eqan 
r_{\phi}(u_y^{(1)} \ldots u_y^{(A)}|u_1,u_2 )  \equiv 
\prod_{a=1}^A  r_{\phi}(u_y^{(a)}|u_1,u_2 ) \,.
\enan 
The parametrization we use for the inference distribution $r_{\phi}(u_y|u_1,u_2 )$ is detailed below in~\Cref{app:inference}. We are interested in 
minimizing the r.h.s. of (\ref{app_eq:I_q_bound}) w.r.t. $(\theta,\phi)$.
To obtain low-variance gradients, 
it is convenient to eliminate the $\theta, \phi$ dependence 
in the measures of (\ref{eq:B1})-(\ref{eq:Dk}) using the `reparametrization trick'~\cite{kingma2013auto}. 

The idea is to obtain samples from 
 $c_{\theta}(u_y,u_1,u_2)$ by a $\theta$-dependent transformation of 
three $\textrm{Unif} [0,1]$ samples $\vv=(v_y,v_1,v_2)$,
and samples from $r_{\phi}(u_y|u_1,u_2)$ 
by a $(u_1,u_2,\phi)$-dependent transformation of 
$\epsilon \sim \textrm{Unif}[0,1]$.  
We present the details of these transformations
in Appendices~\ref{app:inverse_CDF} and ~\ref{app:inference},
respectively.

Taking $M$ samples of $c_{\theta}(u_y,u_1,u_2)$
and denoting them 
as $\bar{u}_y^{(m)}, \bar{u}_1^{(m)}, \bar{u}_2^{(m)}$,
we can estimate (\ref{eq:B1}) as
\eqan 
B_1[\theta] \simeq
\frac{1}{M} \sum_{m=1}^M 
\log \left[ c(\bar{u}^{(m)}_y, \bar{u}^{(m)}_1) \, 
c_{1,2|\theta \left(\bar{u}_y^{(m)} \right)}( \bar{u}^{(m)}_{1|y}, \bar{u}^{(m)}_{2|y}))
\right] 
\enan 
where we denoted $\bar{u}_{i|y} = F(u_i= \bar{u}_i| u_y=\bar{u}_y )$ for $i=1,2$. 
An estimate of the gradient $\nabla_{\theta}B_1$
is obtained by acting on this  expression with $\nabla_{\theta}$, which 
also acts on the $\theta$-dependent samples.

Denoting $A$  samples from 
$r_{\phi}(u_y|u_1,u_2)$ as 
~$\hat{u}_{y}^{(a)}$, we can also estimate  (\ref{eq:B2}) as 
\eqan 
B_2[\theta,\phi] 
 \simeq 
-\frac{1}{M}\sum_{m=1}^M
\log \left(
\frac{1}{K}\sum_{a=1}^A 
w_{a,m}
\right) \,,
\label{eq:B2_estimate}
\enan 
where we defined
\eqan
w_{a,m} = 
\frac{c_{\theta}(\hat{u}_{y}^{(a)}, 
\bar{u}^{(m)}_1, \bar{u}^{(m)}_2)}
{r_{\phi}(\hat{u}_{y}^{(a)}|\bar{u}^{(m)}_1, \bar{u}^{(m)}_2) } \,.
\enan 
Acting on this expressions with $\nabla_{\theta}$ yields 
an estimate of  $\nabla_{\theta}B_2$. 
On the other hand, as noted in~\cite{rainforth2018tighter}, 
the estimate of 
$\nabla_{\phi}B_2$ resulting from acting with 
$\nabla_{\phi}$ on~(\ref{eq:B2_estimate})
has a signal-to-noise ratio which decreases with $A$.
A solution to this problem was found in~\cite{tucker2018doubly}, which 
showed that a stable gradient estimate can be obtained instead as
\eqan 
\nabla_{\phi}B_2 \simeq \frac{-1}{M}
\sum_{m=1}^M\sum_{a=1}^A
 \left( \frac{w_{a,m}}{\sum_{s=1}^A w_{s,m}} \right)^2
 \frac{\partial \log w_{a,m}}{\partial \hat{u}_{y}^{(a)}}
  \nabla_{\phi} \hat{u}_{y}^{(a)} \,, 
\enan 
and this is the estimate 
we  use in our experiments.

\section{The bivariate Gaussian copula} 
\label{app:gaussian_copula}
A  bivariate Gaussian copula  is parametrized by $\theta \in [-1,1]$ and given by 
\eqan 
c(u_1,u_2) = \frac{1}{\sqrt{1-\theta^2}} 
\exp \left\{  -\frac{\theta^2 (x_1^2 + x_2^2) - 2\theta x_1 x_2}{2(1-\theta^2)}  \right\} 
\enan 
where $x_i = \Phi^{-1}(u_i)$  and $\Phi$ is the standard univariate Gaussian CDF. 
For explicit expressions of other popular bivariate copulas, see~\cite{aas2009pair}. 

\section{Sampling from the copula}
\label{app:inverse_CDF}

In this section we show how to obtain samples from the three-dimensional copula
\begin{eqnarray}
\label{eq:q_expansion2}
\begin{aligned}
c_{\theta}(u_y,u_1,u_2) & =c(u_y, u_1) \, c(u_y, u_2) \, 
  c_{1,2|y,\theta}(u_{1|y}, u_{2|y}) 
\end{aligned}
\end{eqnarray}
by applying a $\theta$-dependent transformation to samples from~Unif$[0,1]$. 
We use the Rosenblatt  transform~\cite{rosenblatt1952remarks}, which consists in using the inverse CDF method
to sample from each factor in 
\eqan 
c(u_y,u_1,u_2) = c(u_1)c(u_y|u_1)c(u_2|u_y,u_1) .
\enan 
We denote $F(\cdot|\cdot)$ is the CDF of $c(\cdot|\cdot)$. 
Adopting the notation of~\cite{aas2009pair}, we define
\eqan 
h_{ij}(u_i,u_j) = F(u_i|u_j) = \frac{\partial C(u_i,u_j)}{\partial u_j }  \,\, i,j=1,2.
\enan 
For several popular parametric families of bivariate
copulas, such as those we consider in this paper, explicit expressions 
are known for $h_{ij}(u_i,u_j)$ along with its inverse 
 $h^{-1}_{ij}(\cdot,u_j)$ w.r.t. the first argument (see e.g.~\cite{aas2009pair}). 
 Note that using this notation, the arguments in the last factor of (\ref{eq:q_expansion2}) are $u_{i|y} = h_{iy}(u_i,u_y)$ (i=1,2).

We first sample $(v_1,v_y,v_2)$ from Unif$[0,1]$ and successively obtain 
$u_1,u_y,u_2$ by inverting the functions in the r.h.s. of 
\begin{eqnarray*}
v_1 &=& F(u_1) \,,
\\
&=& u_1 \,,
\\
v_y &=& F(u_y|u_1) \,,
\\
&=& h_{y1}(u_y,u_1) \,, 
\\
v_2 &=& F(u_2|u_1,u_y) \,, 
\\
&=& h_{21|\theta(u_y)}(F(u_2|u_y),F(u_1|u_y)) \,,
\\
&=& 
h_{21|\theta(u_y)}(h_{2y}(u_2,u_y),h_{1y}(u_1,u_y)) \,.
\end{eqnarray*}
Explicitly, we get 
\begin{eqnarray*}
u_1 &=& v_1 \,,
\\
u_y &=& h_{y1}^{-1}(v_y,u_1) \,,
\\
u_2 &=& h_{2y}^{-1}(h_{21|\theta(u_y)}^{-1}(v_2,h_{y1}(u_y,u_1)),u_y ) \,.
\end{eqnarray*}

Note that only
$u_2$  actually depends on $\theta$.

\section{Parametrization of the learned models}
\label{app:inference} 

\subsection{Parametrizing the learned conditional copula}
The conditional Gaussian copula 
$c_{1,2|\theta(u_y)}(u_{1|y}, u_{2|y})$ is parametrized 
by the function $\theta(u_y):[0,1] \rightarrow [-1,+1]$. For its functional form we used
\eqan 
\theta(u_y) = \tanh \left( \sum_{i=1}^{16} w_{2,i} \tanh (w_{1,i} u_y + b_1)     +b_2\right) 
\enan 
where $w_{1,i},w_{2,i},b_1,b_2 \in \mathbb{R}$. 

\subsection{Parametrizing and sampling from the inference distribution}

In our experiments we parametrize
the inference distribution $r_{\phi}(u_y|u_1,u_2)$
via its CDF, as 
\eqan 
R_{\phi}(u_y|u_1,u_2) &=& \int_{0}^{u_y} du \, r_{\phi}(u|u_1,u_2)  \,, 
\\
&=&
\frac{ 1 }{1 + e^{-z(u_y) a_{\phi}(u_1,u_2) - b_{\phi}(u_1,u_2)  }},
\nn
\enan 
where $z(u_y) = \log \left(\frac{u_y}{1-u_y} \right)$. Derivating w.r.t. $u_y$ gives
\eqan 
r_{\phi}(u_y|u_1,u_2) = R_{\phi}(1-R_{\phi})a_{\phi} 
\left( u_y^{-1} + (1-u_y)^{-1} \right) \,.
\enan

The functions $a_{\phi}(u_1,u_2)$ and $b_{\phi}(u_1,u_2)$ take values in $\mathbb{R}$ and 
are parametrized with a neural network with two hidden layers, and we impose $a_{\phi}(u_1,u_2) >0$ in order to make $R_{\phi}$ monotonous with $u_y$.
In order to sample from $r_{\phi}$, we draw $\epsilon \sim \textrm{Unif}[0,1]$, and use 
the inverse CDF method to obtain
\eqan 
u_y(\epsilon, u_1,u_2) &=& R^{-1}_{\phi}(\epsilon|u_1,u_2) \,, 
\\
&=& \frac{ 1 }{1 + e^{-(z(\epsilon) - b_{\phi}(u_1,u_2))/a_{\phi}(u_1,u_2)   }} \,.
\nn
\enan

\section{Comparison with a discrete estimator}
In this section we estimate the PID of the two models of three neurons from eq.~(\ref{eq:models12}) using the discrete estimator~\texttt{BROJA-2PID}~\cite{makkeh2018broja}. 
In particular, we present a quantization scheme of the continuous models that leads to a qualitative agreement between the discrete and continuous estimators, thus further validating  the results of the latter. 

Let us denote the discretized versions of $X_1, X_2, Y$ as 
$\hat{x}_1, \hat{x}_2, \hat{y}$. 
The discrete PID estimators require as input 
a distribution $p(\hat{x}_1,\hat{x}_2,\hat{y})$~\cite{makkeh2018broja}. To create the latter from our continuous models in eq.~(\ref{eq:models12}), we start by dividing the continuous range of each $X_i (i=1,2)$ into~$N_x$ segments, and associate each segment with a discrete  value $\hat{x}_i$ 
equal to the value of $X_i$  in the middle of each segment. 
To each square in the resulting 2D $N_x \times N_x$ grid 
we associate a discrete probability $p(\hat{x}_1,\hat{x}_2)$
equal to the integral of the joint Gaussian density 
of $(X_1,X_2)$ in the square. 
Finally, in each of the two models, we split the $Y$ range into $N_y$ segments $\{ s_i\}_{i=1}^{N_y}$. The boundaries of the segments 
are chosen such that the same fraction $1/N_y$ of values of $Y=Y(X_1,X_2)$ falls into each segment using eq.~(\ref{eq:models12}), a procedure called `maximum entropy binning'. 
Let~$\hat{y} \in \{1 \ldots N_y \}$. 
Using this  quantization, the three-dimensional discrete distribution is defined as 
\begin{align}
p(\hat{x}_1,\hat{x}_2,\hat{y})
& = 
\begin{cases}
p(\hat{x}_1,\hat{x}_2) \quad &  \textrm{if} \,\,\,  Y(\hat{x}_1,\hat{x}_2) \in s_{\hat{y}} \,,
\\
0 \quad &  \textrm{if} \,\,\,  Y(\hat{x}_1,\hat{x}_2) \notin s_{\hat{y}} \,,
\end{cases}
\end{align} 
where in each model 
$Y(\hat{x}_1,\hat{x}_2)$ is obtained from eq.~(\ref{eq:models12}). 
Fig.~\ref{fig:discrete_model_1_model_2} shows the results of 
the discrete PID obtained using this quantization for the two models considered, assuming $X_i \in [-8, 8]$ and $N_x = 16$ equally-sized segments. For the $Y$ quantization we used $N_y=3$. Note the qualitative agreement with the continuous results in~Fig.~\ref{fig:PID}. 

\begin{figure}[h!]
    \centering
    \includegraphics[width=.95\textwidth]{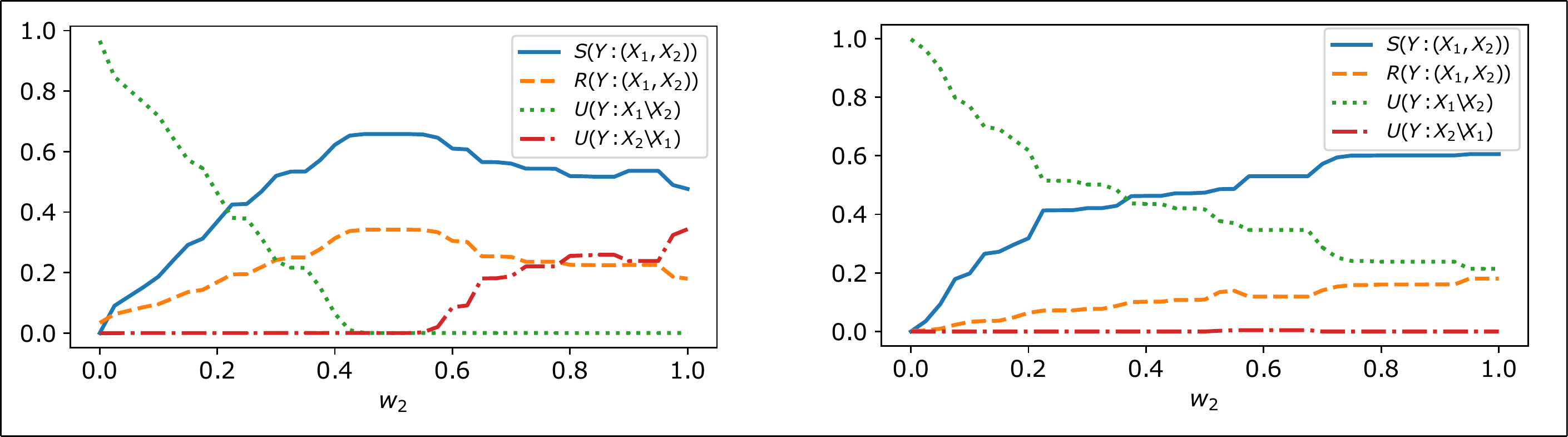}
    \caption{{\bf Qualitative agreement of discrete and continuous PID estimations.} The two models are defined in~\eqref{eq:models12}
    (left: Model 1, right: Model 2), and we used the same model
    parameters  indicated in Fig.~\ref{fig:PID}.
We show the normalized discrete PID terms as a function of the synaptic strength~$w_2$. See the text for details on the discrete quantization used. 
     Note that for both models the discrete results agree qualitatively with the continuous results in Fig.~\ref{fig:PID}.}
    \label{fig:discrete_model_1_model_2}
\end{figure}





\section{Consistency}
Using Model 2 from~Eq.(\ref{eq:models12}) 
as an example, 
we compared estimates of $U(Y \rt X_2 \backslash X_1)$ 
with indirect estimates obtained from applying the consistency conditions to estimates of~$U(Y \rt X_1 \backslash X_2)$.
The results in~\cref{fig:consistency} show good agreement, thus further validating the method.
\begin{figure}[h!]
\begin{center}
{
\includegraphics[width=.3\textwidth]{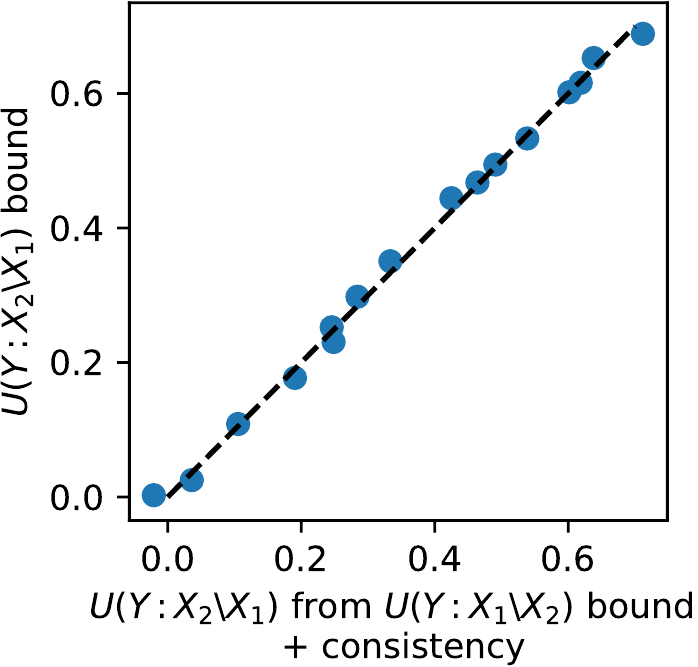}
}
\end{center}
\caption{Comparison of direct vs. indirect estimates of $U(Y \rt X_2 \backslash X_1)$, 
illustrating the consistency of the method. 
}
\label{fig:consistency}
\end{figure}

\section{More on the experiments\label{apx:experiments}}
In this section we provide more details on the last two experiments presented in~\Cref{sec:examples}. 

{\bf Computational aspects of connectivity in recurrent neural circuits.}
\\
We start by deriving the  relation  $TE = S + U_1$ verified in this experiment (Fig.~\ref{fig:conn}d). Transfer entropy~\cite{schreiber2000measuring} $TE(X\to Y)$ is defined as $I(Y^+ : X^- \mid Y^-)$ where $Y^+$ is the future of state of $Y$, $X^-$ and $Y^-$ are the past states of $X$ and $Y$, respectively.
Consider the chain rule for mutual information,
\eqan 
I(Y^+ \rt (X^-,Y^-)) = I(Y^+ \rt Y^-)  + I(Y^+ \rt X^-|Y^-) \,.
\enan 
Replacing $I(Y^+ \rt (X^-,Y^-))$ and $I(Y^+ \rt Y^-)$ by 
the r.h.s. of (\ref{eq:consistency1}) and (\ref{eq:consistency3}), we get \eqan 
I(Y^+ \rt X^-|Y^-) = U(Y^+ \rt X^- \backslash Y^-) + S(Y^+ \rt (X^-, Y^-)) \,,
\enan 
which is the equation we verified by estimating  separately the left and right sides. The two terms in the r.h.s. are called 
\emph{state-independent transfer entropy} and 
\emph{state-dependent transfer entropy} respectively in~\cite{williams2011generalized}, reflecting their intuitive meaning.

In Fig.~\ref{fig:sup1}, we analyze the state space of the network in Fig.~\ref{fig:conn} of the main text. The activities of the upstream sub-network X and downstream sub-network Y are shown, projected onto their first 
two principal components (PCs). The causal structure and algorithmic details of the effective connectivity between the two sub-networks cannot be identified solely by the observation of their geometrical properties.

{\bf Uncovering a plurality of computational strategies in RNNs trained to solve complex tasks.}
\\
{Each RNN has fully connected architecture with \texttt{tanh} non-linearity. Data was generated by sampling from the GMM with K components ($K \in \{4,6,8,10\}$) in batches of 128 data points with the total number of $3000$ batches. The RNNs were trained using standard backprop in time using \texttt{Adam} optimizer in \texttt{Pytorch} package with a learning rate of $0.01$. For each trained RNN we considered all triplets $(Y,X_i,X_j)$ where $i,j \in \{1,\dots,5\}$, i.e. the target variable is the output of the network $Y$ and the source variables iterate over all pairs of the hidden nodes in the RNN. Once the RNN is trained we collect a test sample of $1000$ data points from the same GMM used for training, and evaluate the nodes when inputting the RNN using test data and running it forward for $t=10$ time steps. This gives us $1000$ samples from each variable $X_{1:5},Y$ which we then use for PID analysis on the triplets mentioned above. For each level of task difficulty $K \in \{4,6,8,10\}$ we trained 5 RNNs and performed PID ($A=100$) on the resulting trained networks.}

In Fig.~\ref{fig:sup2} more details on the trained RNN's in Fig.~\ref{fig:xor} of the main text are illustrated, providing more insight into the computational strategies employed by each trained instance as the task complexity grows. The first row shows the time evolution of the recurrent layer of hidden units projected onto their first 3 PC's. For these RNN instances, the ones with $K=6,10$ have grand mother-like cells (as confirmed by the receptive field plots in Fig.~\ref{fig:sup2}c), with large unique information compared to the other cells. These grand mother-like cells cannot be inferred by just inspecting the geometry of the hidden units in the state space, but can be identified with the PID. PID reveals more details about the computation and the differences between strategies for different instances of trained RNN's. {Details of the PID for individual hidden nodes including average unique and synergistic information for each node, its mutual information with the output node, and the decoding weight connecting the hidden node to the output unit is included in Table~\ref{tab:neurons-xor}.}

\begin{figure*}[t!]
	\begin{center}\fbox{\includegraphics[width=1\textwidth]{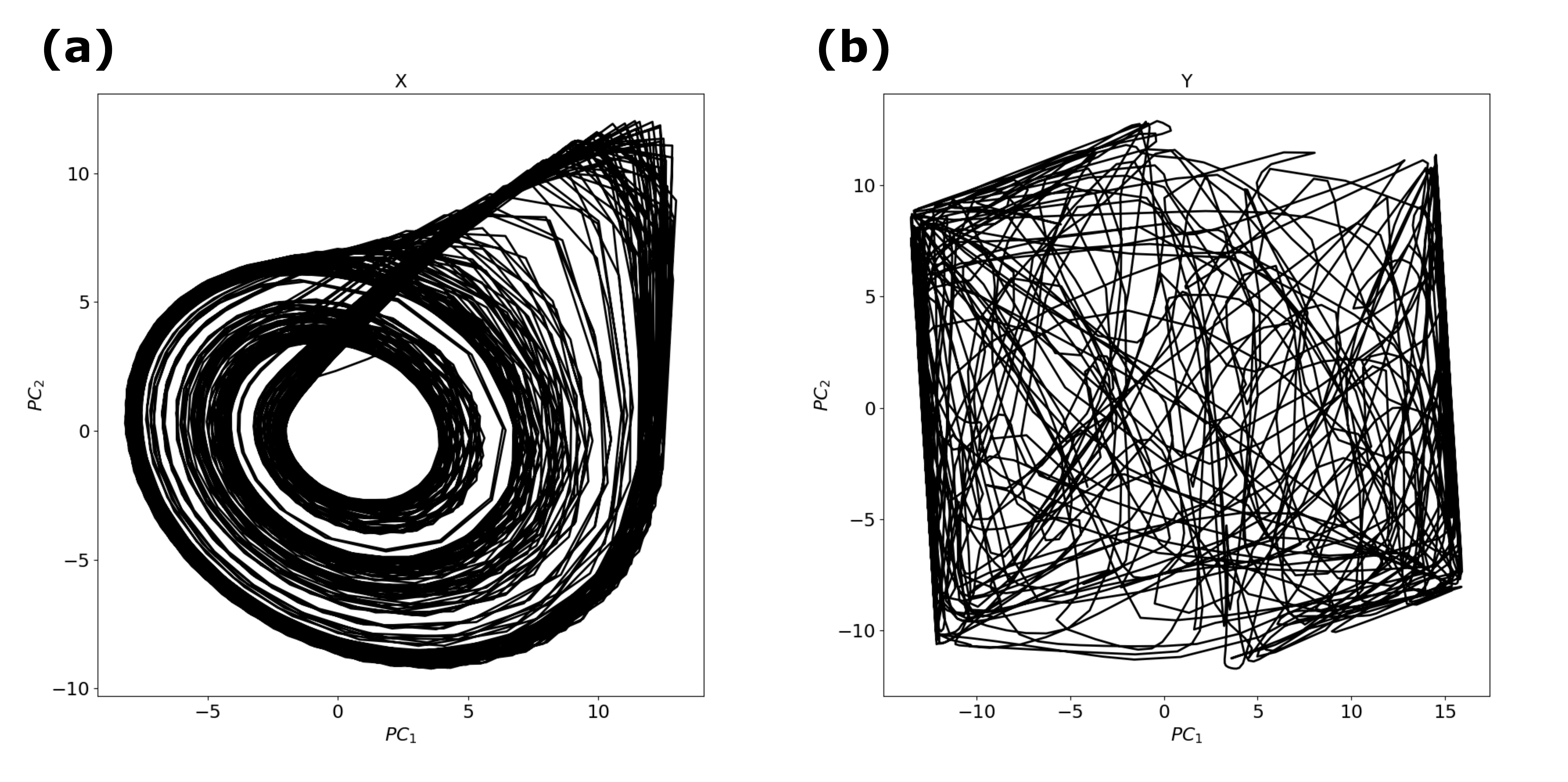}}
	\end{center}
	\caption{ {\bf State space of the chaotic network of rate neurons:} Projection of the state space of the recurrent units for upstream network $X$ (a) and downstream network $Y$ (b) onto their respective first two principal components.}
	\label{fig:sup1}	
\end{figure*}

\begin{figure*}[t!]
	\begin{center}
	\fbox{\includegraphics[width=.95\textwidth]{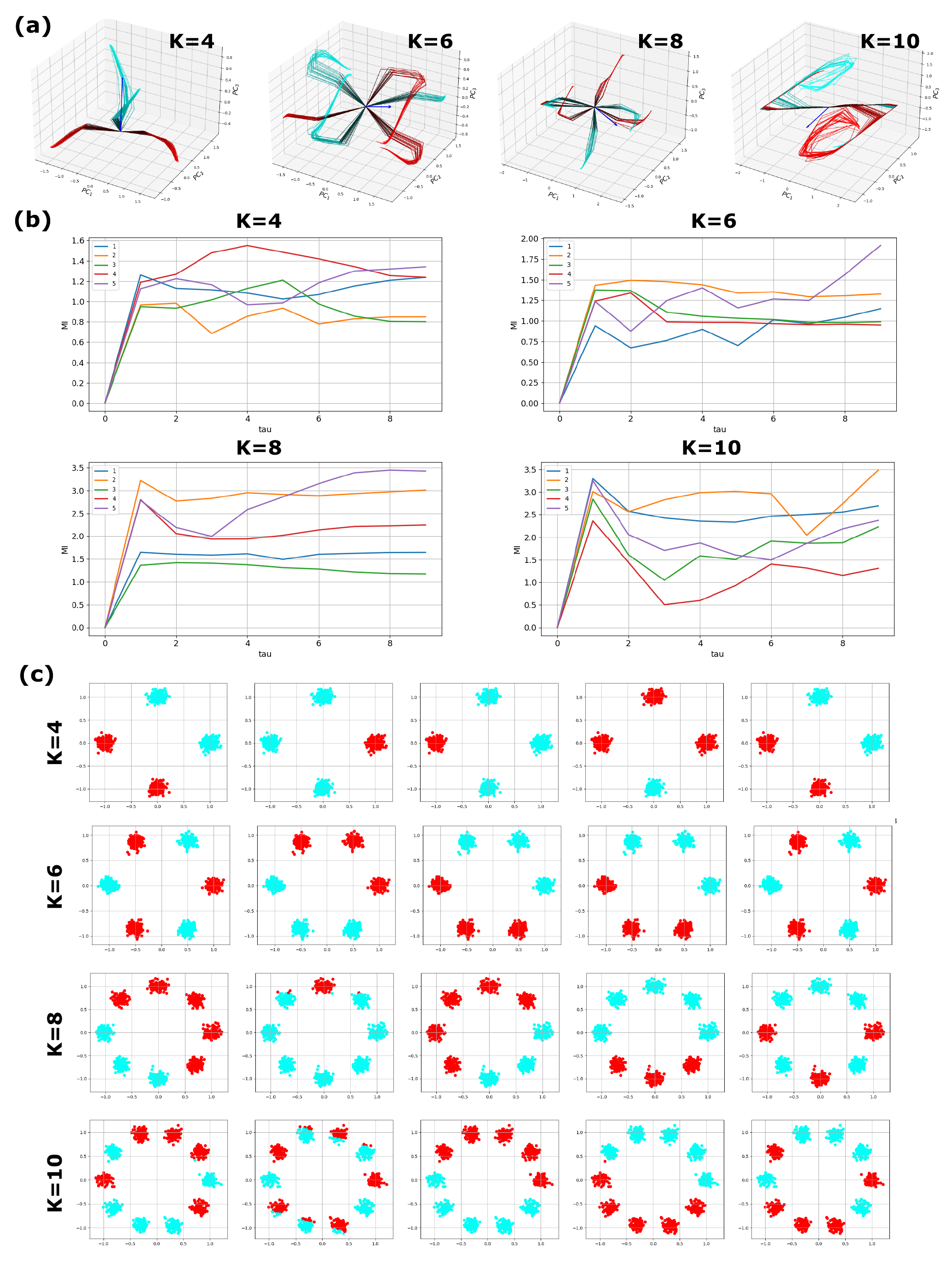}}\end{center}
	\caption{{\bf Algorithmic investigation of trained RNN's on generalized XOR task:} (a) Evolution of the hidden unit activations in time (recurrent time steps). Darker colors correspond to earlier time points; red and cyan correspond to even and odd trials. Blue arrow corresponds to decoding direction, i.e. the predicted label is given by the sign of the projection of the last time point of each trajectory onto this direction. (b) Mutual information between individual hidden units and the output of the network as a function of recurrent time steps for the different tasks. (c) Receptive fields of individual neurons, in certain cases (K=6, unit 1 and K=10, unit 2) grand mother-like cells can be observed, yielding greater unique information than synergistic information hinting at the algorithmic strategy employed by that instance of the trained RNN.
	}
	\label{fig:sup2}	
\end{figure*}


{\color{red}
\begin{table}[]
\resizebox{\textwidth}{!}{
\begin{tabular}{|c|c|c|c|c|l|c|c|c|c|l|c|c|c|c|l|c|c|c|c|}
\hline
\multicolumn{1}{|l|}{} & \multicolumn{4}{c|}{\textbf{K=4}}                    &                            & \multicolumn{4}{c|}{\textbf{K=6}}                    &           & \multicolumn{4}{c|}{\textbf{K=8}}                    &                            & \multicolumn{4}{c|}{\textbf{K=10}}                   \\ \hline
\textbf{N}             & \textbf{UI} & \textbf{SI} & \textbf{MI} & \textbf{W} & \multirow{6}{*}{\textbf{}} & \textbf{UI} & \textbf{SI} & \textbf{MI} & \textbf{W} & \textbf{} & \textbf{UI} & \textbf{SI} & \textbf{MI} & \textbf{W} & \multirow{6}{*}{\textbf{}} & \textbf{UI} & \textbf{SI} & \textbf{MI} & \textbf{W} \\ \cline{1-5} \cline{7-15} \cline{17-20} 
1                      & 0.28        & 0.27        & 1.22        & -0.05      &                            & 0.72        & 0.38        & 1.1         & -0.11      &           & 0           & 0.22        & 1.59        & 0.09       &                            & 0.05        & 0.19        & 2.7         & 0.27       \\ \cline{1-5} \cline{7-15} \cline{17-20} 
2                      & 0.39        & 0.59        & 0.9         & -0.57      &                            & 0.01        & 0.39        & 1.24        & -0.15      &           & 1.58        & 0.42        & 2.93        & -0.57      &                            & 1.74        & 0.13        & 3.47        & -0.55      \\ \cline{1-5} \cline{7-15} \cline{17-20} 
3                      & 0.04        & 0.39        & 0.84        & 0.18       &                            & 0.01        & 0.39        & 0.93        & -0.14      &           & 0.04        & 0.05        & 1.18        & 0.07       &                            & 0.38        & 0.46        & 2.21        & 0.04       \\ \cline{1-5} \cline{7-15} \cline{17-20} 
4                      & 0.01        & 0.22        & 1.2         & -0.38      &                            & 0.02        & 0.32        & 0.93        & 0.02       &           & 0.34        & 0.33        & 2.16        & 0.03       &                            & 0           & 0.37        & 1.24        & 0.02       \\ \cline{1-5} \cline{7-15} \cline{17-20} 
5                      & 0.4         & 0.29        & 1.38        & -0.34      &                            & 0.56        & 0.69        & 1.8         & -0.43      &           & 0.01        & 0.41        & 3.39        & -0.62      &                            & 0.02        & 0.2         & 2.32        & -0.21      \\ \hline
\end{tabular}}
\vspace*{2mm}
\caption{\label{tab:neurons-xor}\textbf{Node-specific details for generalized XOR task:} Average node-specific unique, synergistic, and mutual information (UI, SI, MI) and the decoding weight for different nodes in the hidden layer ($N \in \{1,\dots,5\}$) and for different task difficulty levels ($K \in \{4,6,8,10\}$).}
\end{table}
}




\end{document}